\documentclass[aps,pra,onecolumn,nofootinbib,superscriptaddress]{revtex4}
\usepackage{amsmath,amssymb,amsthm}
\usepackage{amsfonts}
\usepackage{amssymb}
\usepackage{graphicx}
\usepackage{comment}
\usepackage{color}
\usepackage[a4paper, total={6in, 8in}]{geometry}

\newtheorem{theorem}{Theorem}
\newtheorem{definition}[theorem]{Definition}
\newtheorem{corollary}[theorem]{Corollary}

\newtheorem{lemma}[theorem]{Lemma}
\usepackage[title]{appendix}

\begin{document}

\title{Mutually Unbiased Unitary Bases of Operators on $d$-dimensional Hilbert Space}

\author{Rinie N. M. Nasir}
\affiliation{Faculty of Science, International Islamic University Malaysia (IIUM),
Jalan Sultan Ahmad Shah, Bandar Indera Mahkota, 25200 Kuantan, Pahang, Malaysia}
\author{Jesni Shamsul Shaari}
\affiliation{Faculty of Science, International Islamic University Malaysia (IIUM),
Jalan Sultan Ahmad Shah, Bandar Indera Mahkota, 25200 Kuantan, Pahang, Malaysia}
\affiliation{ Institute of Mathematical Research (INSPEM), University Putra Malaysia, 43400 UPM Serdang, Selangor, Malaysia.}
\author{Stefano Mancini}
\affiliation{School of Science \& Technology, University of Camerino, I-62032 Camerino, Italy}
\affiliation{ INFN Sezione di Perugia, I-06123 Perugia, Italy}

\date{\today}
\begin{abstract}
Analogous to the notion of mutually unbiased bases for Hilbert spaces, we consider mutually unbiased unitary bases (MUUB) for the space of operators, $M(d,\mathbb{C})$, acting on such Hilbert spaces. The notion of MUUB reflects the equiprobable guesses of unitary operators in one basis of $M(d,\mathbb{C})$ when estimating a unitary operator in another. Though, for prime dimension $d$, the maximal number of MUUBs is known to be $d^2-1$, there is no known recipe for constructing them, assuming they exist. However, one can always construct a minimum of three MUUBs, and the maximal number is approached for very large values of $d$. MUUBs can also exist for some $d$-dimensional subspace of $M(d,\mathbb{C})$ with the maximal number being $d$. 

\end{abstract}

\maketitle



\section{Introduction}	
\par In the context of quantum information theory, one's ability to know or manipulate a system is generally limited. The uncertainty principle as an example, limits the ability to precisely estimate values associated to non-commuting observables. Considering quantum systems that can be represented by elements in a finite dimensional Hilbert space, measurements made in one basis may perturb the system and effectively result in introducing uncertainty of measurements made in another. In the extreme scenario, where a state prepared in one  basis gives maximum uncertainty when measured in another, is captured in the notion of mutually unbiased bases (MUBs). More precisely, measuring a quantum state belonging to a basis along a mutually unbiased basis, one obtains as the result, a random vector of the latter basis and all the possible results are equiprobable \cite{Durt}. As a quick example, consider in a 2-dimensional Hilbert space, a state from the basis, $\{(|0 \rangle +|1 \rangle)\sqrt{2},(|0 \rangle -|1 \rangle)\sqrt{2}\}$, with $|0 \rangle$ and $|1 \rangle $ elements of the computational basis; measurements in this latter basis would result in either the states $|0 \rangle$ or $|1 \rangle$ with probability $1/2$ each. The following is the standard definition for MUB 
\begin{definition}
Two distinct orthonormal bases for a $d$-dimensional Hilbert space, ${ B^{(0)}}=\{ \vert \varphi_0 \rangle,...,\vert \varphi_{d-1} \rangle \}$ and ${ B^{(1)}}=\{ \vert \phi_0 \rangle,...,\vert \phi_{d-1} \rangle \}$ are said to be mutually unbiased bases (MUB) provided that $\vert \langle \varphi_i \vert \phi_j \rangle\vert=1/\sqrt{d}$, for every $i, j=0,...,{d-1}$.
\end{definition}
\noindent MUBs have proven to be useful in practical applications such as quantum key distribution (QKD)\cite{Gisin} and quantum state tomography\cite{Fields}.

\par We consider an analogous idea of equiprobable guesses to bases consisting of unitary operators for some subspace of $M (d,\mathbb{C})$, namely mutually unbiased unitary bases (MUUB)\cite{Scott,Jesni}.\footnote{Although $M (d,\mathbb{C})$ is the set of $d \times d$ matrices with complex entries, it is regarded as the set of operators acting on a $d$-dimensional Hilbert space (with $d$ prime) because actually matrices represent such operators.  } In short, these are the bases such that the Hilbert Schmidt inner product for elements from different bases would have a common value. The idea of equiprobable guesses can be understood as follows \cite{Vidal}; let $\mathcal{V}_{\mathcal{R}}$ be a guess for a given unknown unitary operator, $\mathcal{V} \in SU(d)$ randomly selected based on a probability distribution uniform with respect to the Haar measure. How closely $\mathcal{V}_{\mathcal{R}}$ resembles $\mathcal{V}$ (when transforming any state $| \upsilon \rangle \in \mathcal{H}_d$) is defined by a function, $G(\mathcal{V}_{\mathcal{R}},\mathcal{V})$  given by  
\begin{equation}
G(\mathcal{V}_{\mathcal{R}},\mathcal{V})=\left| \int \langle \upsilon |\mathcal{V}_{\mathcal{R}}^{\dagger} \mathcal{V} | \upsilon \rangle d \upsilon\right|^2=\frac{1}{d^2} |\text{Tr} (\mathcal{V} \mathcal{V}_{\mathcal{R}}^{\dagger})|^2.
\end{equation}
\noindent
Note that this represents the average over all pure states $| \upsilon \rangle$. Hence, two distinct guesses for $\mathcal{V}$, say, $\mathcal{V}_{\mathcal{R}}$ and $\mathcal{V}_{\mathcal{S}}$ are equiprobable if $|\text{Tr} (\mathcal{V} \mathcal{V}_{\mathcal{R}}^{\dagger})|^2=|\text{Tr} (\mathcal{V} \mathcal{V}_{\mathcal{S}}^{\dagger})|^2$. The notion of MUUB (for the entire space of $M (d,\mathbb{C})$) was first mentioned in Ref. \cite{Scott} and was used to construct a unitary-2 design relevant to process tomography of unital quantum channels. A complete set of of MUUBs for the space $M (d,\mathbb{C})$ can at most have $d^2-1$ number of bases and the same work provided a construction for $d=2, 3,5,7$ and $11$. This definition was further generalized in Ref. \cite{Jesni} to include MUUB for some subspace of $M (d,\mathbb{C})$. The definition for MUUBs is as follows, 
\noindent
\begin{definition}\label{defMUUB}
Consider some $n$ dimensional subspace, $\mathcal{S}$, of the vector space $M (d,\mathbb{C})$. Two distinct orthogonal bases of $\mathcal{S}$ composing of unitary transformations, $\mathcal{F}^{(0)}=\{f_0^{(0)},...,f_{n-1}^{(0)}\}$ and $\mathcal{F}^{(1)}=\{f_{0}^{(1)},...,f_{n-1}^{(1)}\}$, are MUUB provided that,
\begin{eqnarray}
\vert \text{Tr} (f_i^{(0)^\dagger} f_j^{(1)}) \vert ^2=C,~\forall f_i^{(0)} \in \mathcal{F}^{(0)}, f_j^{(1)} \in \mathcal{F}^{(1)}
\end{eqnarray}
for $i, j=0,...,n-1$ with $\mathcal{F}^{(0)}$ and $\mathcal{F}^{(1)} \in  \mathcal{S}$ and some constant $C \neq 0$.
\end{definition}


\noindent Beyond its role (as described in Ref. \cite{Scott}) in ancilla assisted quantum process tomography, MUUB has also found use in quantum cryptography\cite{jspla,js}. These are bidirectional quantum key distribution schemes where encoding is done by using unitary operators from differing MUUBs.

One can imagine that, akin to the construction of MUBs, the search and construction of MUUBs is a nontrivial task. While the study of MUBs has received much attention and has been understood to quite an extent (at the very least, constructions for the maximal number of MUBs for Hilbert spaces of prime powered dimensions has been established\cite{Mancini}) the study of MUUBs is very much at its infancy. We will review to a certain extent in this paper, the current understanding of MUUBs mainly based on Refs. \cite{Jesni, Rinie}. Representing new results in this work, we provide the minimal number of MUUBs for $M(d,\mathbb{C})$ and show that it approaches the maximal possible number for very large $d$. This is done by viewing the problem of construction of MUUBs as equivalent to constructing MUBs containing only maximally entangled states (MES); given the isomorphism between unitary operators and maximally entangled states. Ref. \cite{Wiesniak} shows the amount of entanglement is conserved in constructing a complete set of MUBs for bipartite states and thus some related results thereof applies directly to our case. It should be noted that the search for MUBs consisting of MES has its own interest. \cite{Tao,Liu,Dengming}  

A brief outline of the work is as follows.We begin in Sec. 2 with a straightforward numerical approach for constructing MUUB. This is mainly derived from Refs. \cite{Jesni}. Sec. 3 deals with the issue of the minimal number of MUUB that can exist. In Sec. 4, following the work of Ref. \cite{Rinie}, we consider the maximal number of MUUBs for some $d$-dimensional subspace of $M(d,\mathbb{C})$ based on an isomorphism between the monoids defined for the underlying sets $\mathcal{H}_d$ and that for the subspace of $M(d,\mathbb{C})$. 

\section{Numerical search for MUUB}
\noindent
For $d^2$ dimensional space of $M(d,\mathbb{C})$, the maximal and minimal number of MUUBs for a subspace of $M(d,\mathbb{C})$ are $d^{2}-1$ and $\max[3,d(d-1)]$ respectively\footnote{a smaller number for the minimal case was given in Ref. \cite{Jesni}}. The maximal number of MUUBs for dimensionality of $d^2$ was derived in Ref. \cite{Scott}. 

The minimal number on the other hand, can be derived using a relevant lemma in Ref. \cite{Wiesniak} describing the maximal number of mutually unbiased bases consisting of product states given a complete set of mutually unbiased bases of a bipartite systems. This will be further explained in the next subsection. 

In the absence of a proper method to construct the complete set of MUUBs, we consider a straightforward numerical approach. Consider the $d^2$-dimension space of $M(d,\mathbb{C})$ with the \textit{canonical} basis $\mathcal{K}$ given by
\begin{equation}\label{E0}
\{ X^{r} Z^{s} | r \in [0,d-1], s \in [0,d-1] \}
\end{equation}
\noindent
with $X$ and $Z$ as the generalized Pauli operators \cite{Hall}.  
\noindent
Define another basis, $\mathcal{L}$, given by
\begin{equation}\label{muub}
\mathcal{L}=\{ \mathcal{E}_{mn}= Y X^{m} Z^{n} | m \in [0,a], n \in [0,b]    \}
\end{equation}
\noindent
where $ Y \in \mathcal{L}$ is written as
\begin{equation}
Y=\sum_r \sum_s {\gamma}_{rs} X^{r} Z^{s}
\end{equation}
\noindent
for $r, s \in [0,d-1]$ and for every element $X^{r} Z^{s} \in \mathcal{K}$, ${\gamma}_{rs}\in \mathbb{C}$. It is obvious to note that the elements in $\mathcal{L}$ are orthogonal to each other given that $ |\text{Tr} [(\mathcal{E}_{mn})^{\dagger} \mathcal{E}_{m'n'}] |=|\text{Tr} [(X^{m}Z^n)^{\dagger} Y^\dagger Y X^{m'} Z^{n'} ]|=|\text{Tr}  [(X^{m} Z^{n})^{\dagger}  (X^{m'} Z^{n'})] |=0$. 
The following two conditions must then be met:
\begin{enumerate}
\item $\mathcal{E}_{mn}^{\dagger} \mathcal{E}_{mn}=\mathbb{I}_d$ (unitary condition)
\item $\forall m,n,{ r,s}, | \text{Tr}~ \mathcal{E}_{mn}^{\dagger} X^{{r}} Z^{{s}} |^2=1$ (mutually unbiased condition with the constant $C$ of Definition \ref{defMUUB} being equal to 1, i.e. the case of Ref. \cite{Scott}.)
\end{enumerate}
\noindent
From the second condition above, we note that $\forall {r,s}, |\text{Tr} (\mathcal{E}_{mn}^{\dagger} X^{{r}} Z^{{s}})|^2=|{\gamma}_{{rs}} \text{Tr} (\mathbb{I}_d)|^2$ and $|\sum_{{rs}}{\gamma}_{{rs}}\text{Tr} (\mathbb{I}_d)|^2=1$ imply that $\forall { r,s},|{ \gamma}_{{rs}}|^2=1/{d^2}$. We select $\gamma_{{rs}}=\exp(2 \pi i/d)^{g_{{rs}}}$ for some $g_{{rs}} \in \mathbb{Z}_d$. We set the algorithm for the computer to search for a $d^2-1$ element set of $\{g_{{rs}}\}$ given the conditions above. This set $\{g_{{rs}}\}$ together with Eq. (\ref{muub}) provides for a basis mutually unbiased to the standard basis. This is obvious given that, if $Y$ is mutually unbiased to all elements in the standard basis, i.e., $\forall { r,s},|{ \gamma}_{{rs}}|^2=1/{d^2}$, then 
\begin{align}
\forall a,b, ~|\text{Tr} (\mathcal{E}_{mn}^{\dagger} X^{a} Z^{b})|^2=|\text{Tr} [(X^{m}Z^n)^{\dagger}Y^\dagger X^{a} Z^{b} ]|^2\nonumber\\
=|\text{Tr} [\omega_{mn}Y^\dagger X^{a}(X^{m})^{\dagger} Z^{b}(Z^{n})^{\dagger} ]|^2\nonumber\\
=|\text{Tr} [\omega_{mn}\gamma_{rs}\mathbb{I}_d]|^2=1
\end{align}
for some $rs$ and some $|\omega_{mn}|=1$\footnote{this term arises due to the well known identity $XZ=\omega ZX$ where $\omega=\exp{(2\pi i/d)}$}. We then proceed to search for a different set of $\{g_{tu}\}$. We highlight in the following subsection, an example of the numerical search for the construction of the complete 8 MUUBs for dimension $d=3$ as in Ref. \cite{Jesni}. Note that this search to construct MUUB does not confirm the uniqueness of the construction. 

\subsection{Qutrit based example}
\noindent
We begin by letting the canonical basis for $M(3,\mathbb{C})$ to 
be the following;
\begin{equation}\label{S}
\{\mathbb{I}_3, X_3, X_3^2, Z_3, Z_3^2, X_3Z_3, (X_3Z_3)^2, X_3Z_3^2, (X_3Z_3^2)^2\}
\end{equation}
with the generalized Pauli operators of $M(3,\mathbb{C})$, i.e. $X_3$ and $Z_3$ be represented by the following matrices
\begin{equation}
X_3 =
\left(\begin{array}{ccc} 0 & 0 & 1 \\ 1 & 0 & 0 \\ 0 & 1 & 0\end{array}\right)~,~
Z_3 =
\left(\begin{array}{ccc} 1 & 0 & 0 \\ 0 & \gamma_3 & 0\\ 0 &0 & \gamma_{3}^2 \end{array}\right),
\end{equation}

\noindent
where $\gamma_3=\exp ({2 \pi i /3)}$. Note that Eq. (\ref{S}) is a specific example of Eq. (\ref{E0}) for $M(3,\mathbb{C})$. We then proceed with the following steps.\\

{\it Step 1} : Consider a unitary operator, $U_0^{(1)}$ from another basis, say $\mathcal{U}^{(1)}$ given by
\begin{align}\label{E2}
U_0^{(1)}=\frac{1}{\sqrt{9}}[\mathbb{I}_3&+\gamma_3^{t_1} X_3+ \gamma_3^{t_2} X_3^2+ \gamma_3^{t_3} Z_3+\gamma_3^{t_4} Z_3^2\nonumber\\
+&\gamma_3^{t_5} X_3 Z_3+\gamma_3^{t_6} (X_3 Z _3)^2+\gamma_3^{t_7} X_3 Z_3^2+\gamma_3^{t_8} (X_3 Z_3^2)^2]
\end{align}
The superscript (with parentheses) for the unitary operator refers to the basis, while the subscript refers to element in the basis. We now search for the set of index, $\{ t_i~|~ i=1, \cdots, 8,t_i=0,\cdots, d-1\}$ based on the conditions $(U_0^{(1)})^{\dagger} U_0^{(1)}=\mathbb{I}_3$ and $U_{0,0}^{(1)}$ is mutually unbiased to $\{X_3^{a}Z_3^{b} \}$ i.e.,
\begin{align}\label{E4}
&|\text{Tr} ((U_{0,0}^{(1)})^{\dagger} X_3)|^2=|\text{Tr} ((U_{0,0}^{(1)})^{\dagger} X_3^2)|^2=|\text{Tr} ((U_{0,0}^{(1)})^{\dagger} Z_3)|^2\nonumber\\
=&|\text{Tr} ((U_{0,0}^{(1)})^{\dagger} Z_3^2)|^2=|\text{Tr} ((U_{0,0}^{(1)})^{\dagger} X_3 Z_3)|^2|=\text{Tr} ((U_{0,0}^{(1)})^{\dagger} (X_3 Z_3)^2)|^2\nonumber\\
=&|\text{Tr} ((U_{0,0}^{(1)})^{\dagger} X_3 Z_3^2)|^2=|\text{Tr} ((U_{0,0}^{(1)})^{\dagger} (X_3 Z_3^2))|^2=1
\end{align}
\noindent
{\it Step 2}: We determine all other elements of $\mathcal{U}^{(1)}$ as
\begin{equation}\label{E3}
U_{a,b}^{(1)}=\frac{1}{\sqrt{9}}[ (U_{0,0}^{(1)} \cdot X_3^a Z_3^b]
\end{equation}
\noindent
where $a,b=0, \cdots, d^{2}-1$ (excluding the case both $a,b$ equal $0$ as that is the case of  Eq. (\ref{E4})).

We reiterate the above steps for an $l$-th operator of an $k$-th basis$, U_l^{(k)}$ to construct another basis, say $\mathcal{U}^{(k)}$ for a different set $\{t_i\}$. Beyond condition 2, a check should also be made to ensure that the new basis is mutually unbiased to any earlier ones found in case a different set $\{t_i\}$ gives some equivalent set of operators. 
The numerical search thus gives the maximal number of MUUBs (including the standard basis), eight in particular, given by 
\begin{equation}\label{E33}
\mathcal{U}^{(l)}=\{U_{0,0}^{(l)}\cdot X_3^a Z_3^b~|~\forall a,b=0,...,d-1\}
\end{equation}
with  
\begin{align}
U_{0,0}^{(1)}
=& \mathbb{I}_3+\gamma_3 X_3+ \gamma_3^2 X_3^2+ \gamma_3 Z_3+\gamma_3^2 Z_3^2+\gamma_3 X_3 Z_3+\gamma^2 (X_3 Z_3)^2+\gamma_3 X_3 Z_3^2+\gamma_3^2 (X_3 Z_3 ^2)^2\nonumber\\
U_{0,0}^{(2)}=&\mathbb{I}_3+ X_3+ \gamma_3 X_3^2+ Z_3+\gamma_3 Z_3^2+\gamma_3^2 X_3 Z_3+(X_3 Z_3)^2+\gamma_3^2 X_3 Z_3^2+ (X_3 Z_3^2)^2\nonumber\\
U_{0,0}^{(3)}=&\mathbb{I}_3+X_3+ \gamma_3 X_3^2+ Z_3+\gamma_3^2 Z_3^2+ X_3 Z_3+\gamma_3^2 (X_3 Z_3)^2+\gamma_3^2 X_3 Z_3^2+\gamma_3^2 (X_3 Z_3^2)^2\nonumber\\
U_{0,0}^{(4)}=&\mathbb{I}_3+X_3+ \gamma_3 X_3^2+  Z_3+\gamma_3^2 Z_3^2+\gamma_3 X_3 Z_3+ (X_3 Z_3)^2+\gamma_3 X_3 Z_3^2+\gamma_3 (X_3 Z_3^2)^2\nonumber\\
U_{0,0}^{(5)}=&\mathbb{I}_3+ X_3+ \gamma_3^2 X_3^2+  Z_3+\gamma_3 Z_3^2+ X_3 Z_3+\gamma_3^2 (X_3 Z_3)^2+\gamma_3 X_3 Z_3^2+(X_3 Z_3^2)^2\nonumber\\
U_{0,0}^{(6)}=&\mathbb{I}_3+X_3+ \gamma_3^2 X_3^2+  Z_3+\gamma _3 Z_3^2+\gamma_3 X_3 Z_3+ (X_3 Z_3)^2+ X_3 Z_3^2+\gamma_3^2 (X_3 Z_3^2)^2\nonumber\\
U_{0,0}^{(7)}=&\mathbb{I}_3+ X_3+ \gamma_3^2 X_3^2+ Z_3+\gamma_3^2 Z_3^2+\gamma_3^2 X_3 Z_3+\gamma_3^2 (X_3 Z_3)^2+ X_3 Z_3^2+\gamma_3 (X_3 Z_3^2)^2
\end{align}

\noindent
However, such a numerical search is anything but efficient and an algebraic construction is certainly more desirable. 

\section{MUUB and MUBs for MES}
Given the equivalence between unitary operators and MES, the search for MUUBs is equivalent to the search for MUBs consisting of MES.
For a unitary $Y_i$, its equivalent MES, $|\mathcal{Y}_i\rangle\in \mathcal{H}_d\otimes \mathcal{H}_d$, can be written as \cite{sac,d1,d2}, 
\begin{eqnarray}\label{L0}
|\mathcal{Y}_i\rangle=(\sum_{r}\sum_s \langle s|Y_i|r\rangle |r\rangle|s\rangle)/\sqrt{d}~.
\end{eqnarray}
with $|r\rangle$,$|s\rangle$ as some basis vectors for $\mathcal{H}_d$. 

A basis is mutually unbiased with respect to the above could be one with elements provided that $\mathbb{I}_d \otimes  \mathcal{Y}_s |s \rangle \rightarrow |s' \rangle$, where $\langle s | s' \rangle= 1/d$.


Then, a relation has been established between a complete set of MUBs and MES \cite{Andrei,Sych}. Ref. \cite{Wiesniak} explained that one can always find the number of MES for MUBs. Constructions of such maximally entangled basis in arbitrary bipartite system, though not achieving its maximal number, has been studied in Refs. \cite{Liu, Tao} (referred to therein as  mutually unbiased maximally entangled bases or MUMEB). Ref. \cite{Tao} constructed MUMEBs for small $d$ and $d'$ (i.e. for cases of $\mathcal{H}_2 \otimes \mathcal{H}_4$ and $\mathcal{H}_2 \otimes \mathcal{H}_6$), which depends on known Hadamard matrices in bipartite systems with small dimensions. This has its hurdles however as the structure of Hadamard matrices become increasingly more complicated with the increasing of the dimension. Later, Ref. \cite{Liu} constructed MUMEBs by using permutation matrices. To date, Ref. \cite{Dengming} has managed to acquire completely new MUMEBs in bipartite systems with arbitrary dimensions by adding Hadamard matrices which results in constructing $2(d-1)$ MUMEB in $\mathbb{C}^d \otimes \mathbb{C}^d$ for $d \geq 3$. In the following, we will present a very simple approach in constructing MUBs consisting of MES for $\mathcal{H}_d\otimes \mathcal{H}_d$ with a number approaching maximal, $d^2-1$, as the dimensionality, $d$, grows. 
\noindent
\par Given a bipartite state, Ref. \cite{Wiesniak} described that a complete set of MUBs can be formed by either
\begin{enumerate}
\item
product states and MES. 
\item
solely by partially entangled states. 
\end{enumerate}
Note that it is not possible to have a complete set of MUBs built entirely of product states nor entirely of maximally entangled states. In the following, we consider the MES consisting of subsystems $A$ and $B$. We refer to the following lemma from Ref. \cite{Wiesniak} on the number of bases containing only product states, i.e. product bases, 
\begin{lemma}\label{L1}
Let $d_A$ and $d_B$ be the dimensionality of the Hilbert spaces of subsystems $A$ and $B$ respectively. Assume that $d_A  \leq d_B$ . In a complete set of MUBs that contains $d_{A} + 1$ product MUBs, all other bases contain only maximally entangled states.
\end{lemma}

\noindent
This is used to determine the maximal number of product MUBs as described in Ref. \cite{Wiesniak}  (note that $d_A$ and $d_B$ may not necessary be prime numbers). In our case, $d_A=d_B=d$ with $d$ being a prime number. We thus have  the following corollary,
\begin{corollary}\label{C1}
Let $d_A=d_B=d$. Then minimal number of MUBs consisting of MES is $d(d-1)$.
\end{corollary}

\begin{proof}
The proof is rather straightforward. From Lemma \ref{L1}, the maximal number of product bases would be $d+1$. The minimal number of MUBs for MES thus will be
\begin{align}\label{L2}
(d^{2}+1)-(d+1)
=d(d-1)
\end{align} 
\end{proof}

\noindent We note that this will not give the maximal number of MUBs for MES which is $d^{2}-1$ (assuming it exists). However the maximal number is approached for very large values of $d$. Let $R$ be the ratio of the maximal number of MUBs for MES to that of Eq. (\ref{L2})
\begin{eqnarray}\label{L3}
R=\frac{d^{2}-1}{d(d-1)}=\frac{(d+1)(d-1)}{d(d-1)}=\frac{d+1}{d}
\end{eqnarray}

\begin{figure}[h!]
\begin{center}
\includegraphics[angle=0,width=5.0in]{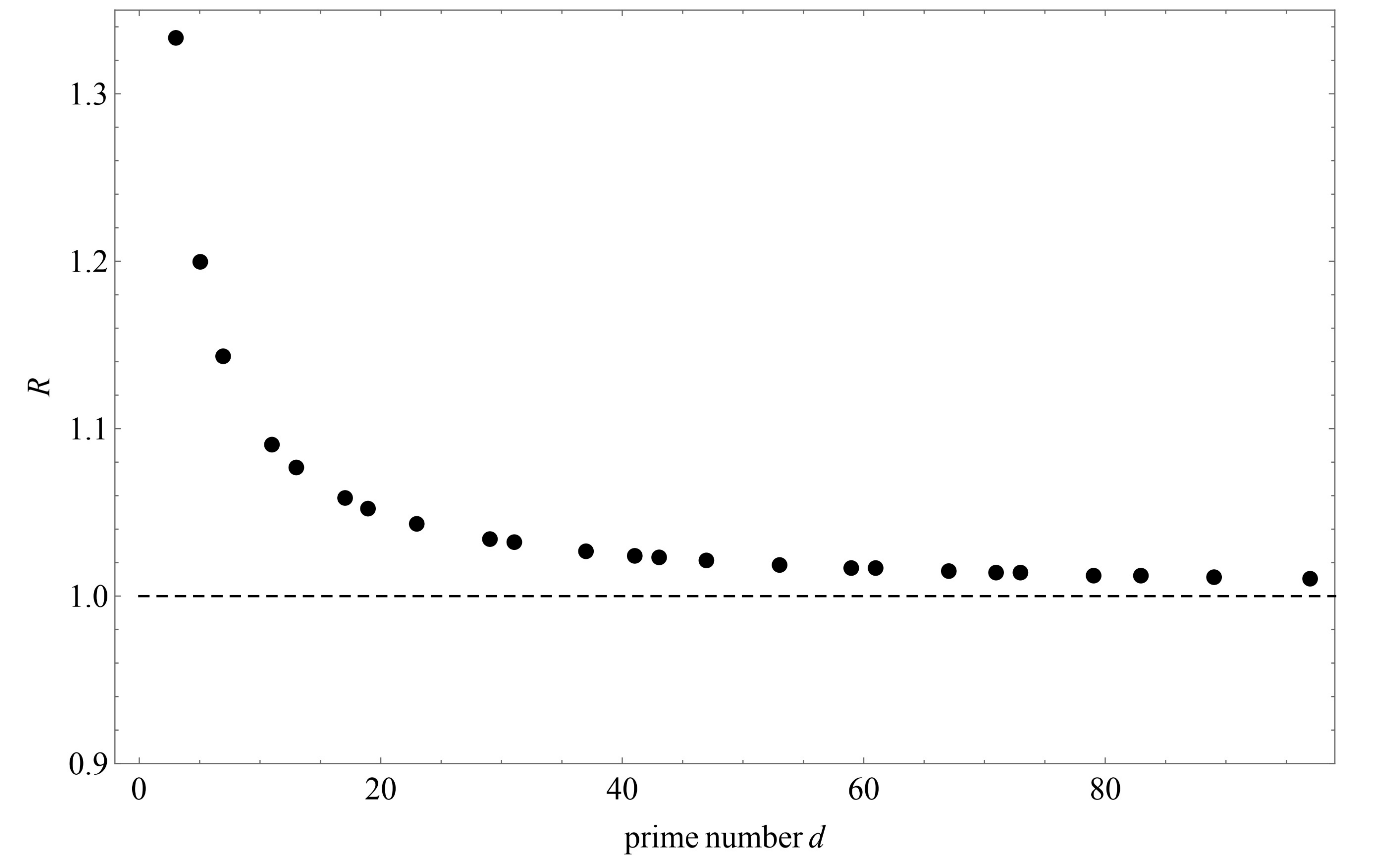}
\end{center}
\vspace{-.5cm}{\textbf{\caption{Quantity $R$ against the prime number $d$ }\label{ratio}}}
\addtocontents{lof}{\protect\addvspace{.5cm}}
\end{figure}
\noindent Figure \ref{ratio} depicts the quantity $R$ against the first 24 prime numbers, $d$ (excluding $d=2$). We can observe that $R$ approaches unity (horizontal dashed line) for very large values of $d$; i.e. $\lim_{d\rightarrow \infty}{R}=\lim_{d\rightarrow \infty}{(d+1)/d}=1$. 

Note that the case $d=2$, Ref. \cite{Jesni} gave the number of MUBs as three. Therefore, coupled with corollary \ref{C1},  we can conclude that for $\mathcal{H}_d \otimes \mathcal{H}_d$ with $d$ being a prime number, one can always construct at least 3 MUBs consisting of MES (or equivalently, 3 MUUBS for the $d^2$-dimensional space of $M(d,\mathbb{C})$). We should stress that, Eq. (\ref{L2}) gives us the minimum number of such MUBs, and exceeds that of  Ref. \cite{Dengming} which gives only $2(d-1)$, namely by a factor of $d/2$. We refer to the details of constructing such MUBs in Ref. \cite{Wiesniak}. The construction itself is somewhat interesting given its simplicity, i.e. it makes use of some entangling (control phase) operator acting on product bases. However, despite the guaranteed existence of such operators, the precise operator to be used for a given $d$ is not dimensional independent nor is it a known function of $d$.

\section{MUUB for $d$-dimensional subspace of $M(d,\mathbb{C})$}
\noindent
Consider a $d$-dimensional subspace of $M(d,\mathbb{C})$, spanned by the basis
\begin{eqnarray}\label{L33}
\{Z^0, Z^1,...,Z^{d-1}\}
\end{eqnarray}
with $Z^0=\mathbb{I}_d$, the identity operator and $Z^i Z^j=Z^{i \oplus j}$. These operators may correspond to the generalised Pauli matrices. In what follows, we shall refer to such a subspace as $\mathcal{K}_{\mathcal{S}}$. A unique isomorphism, $G$, defined by,
\begin{eqnarray}
G(\sum_im_i|i\rangle)=\sum_im_i Z^i~,~\forall a_i\in \mathbb{C}, i=0,...,d-1.
\end{eqnarray}
would map MUBs for $\mathcal{H}_d$ to bases for $\mathcal{K}_{\mathcal{S}}$. Though the absolute value of the inner product on $\mathcal{H}_d$, for any two states in $\mathcal{H}_d$, $\vert \langle \psi \vert \Phi \rangle \vert$, is proportional to the absolute value of the inner product on $\mathcal{K}_{\mathcal{S}}$, this ensures that the bases concerned for $M(d,\mathbb{C})$ would be mutually unbiased to one another. However, this provides no guarantee that the elements in each MUB for $\mathcal{H}_d$ is mapped to unitary operators. As the unitarity of operators is defined based on multiplication between operators, a binary operation on the underlying set of $\mathcal{H}_d$ can be defined accordingly. The resulting monoid is then shown to be isomorphic to the monoid for $\mathcal{K}_{\mathcal{S}}$.
Define the binary operation, $\bullet$, on the set $\mathcal{H}_d$ as follows. 
\begin{definition}
\noindent
For all $\vert \phi \rangle,\vert \psi \rangle \in \mathcal{H}_d $ given as 
\begin{eqnarray}
\vert \phi \rangle=\sum_{i} m_i \vert i \rangle
~,~\vert \psi \rangle=\sum_{j} n_j \vert j \rangle
\end{eqnarray}
with $m_i,n_j\in\mathbb{C}$ and  $i,j=0,...,d-1$, we let $\bullet:\mathcal{H}_d\times \mathcal{H}_d\rightarrow \mathcal{H}_d$ be defined as
\begin{equation}
\sum_{i} m_i \vert i \rangle \bullet \sum_{j} n_j \vert j \rangle = \sum_{q} \alpha_q \vert  q \rangle
\end{equation}
{ where $\alpha_q=\sum_{i,j} m_i n_j$ and $q=i\oplus j$}.
\end{definition}
\noindent It can be shown that $\mathcal{H}_d$ with the operation, $\bullet$, defines a monoid and is isomorphic to the monoid of the subspace of $\mathbb{C}$ defined with multplication\cite{Rinie}. Consider next the following definition. 
\begin{definition}\label{defC}
The conjugate transpose of a state $\vert \phi \rangle=\sum_{i} m_i \vert i \rangle\in \mathcal{H}_d$ is given as $\vert \phi \rangle^\dagger=\sum_{i} m^*_i \vert d-i \rangle$.
\end{definition}
\noindent
This is different from the conventional definition in terms of the dual of a ket, i.e. $|\phi \rangle^\dag = \langle\phi |$. Our definition of such a state is in fact motivated by the conjugate transpose of the operator, say, $K=\sum m_i Z^i$. In other words, if $G^{-1}(K)=\sum_{i} m_i \vert i \rangle=\vert \phi \rangle$, then $G^{-1}(K^\dagger)=\sum_{i} m^*_i \vert d-i \rangle=\vert \phi \rangle^\dagger$. This definition is useful in determining which states in $\mathcal{H}_d$ is mapped to uniatries in the subspace of $\mathbb{C}$. More precisely, $\forall |\phi\rangle\in\mathcal{H}_d,G(|\phi\rangle)$ is a unitary operator if and only if $|\phi\rangle\bullet|\phi\rangle^\dagger=|0\rangle$.


\subsection{Maximal number of MUUB for $\mathcal{K}_{\mathcal{S}}$ }

Before describing the maximal number of MUUBs that can be constructed, we take a quick detour again to the equivalent search for MUBs consisting of MES. Considering only the $d$-dimensional subspace of $\mathcal{H}_d\otimes\mathcal{H}_d$, Ref. \cite{Rinie} subscribes to the same notion of deriving the maximal number of MUBs for a $d$-dimensional space. One needs to consider the space of $d^2-1$ traceless Hermitian operators and MUBs then span subspaces orthogonal to one another. However in considering only MES, the dimensionality of the space of traceless Hermitian operators concerned was shown to be lesser than $d^2-1$, eventually resulting in the maximal number of such MUBs to be lesser than $d+1$. The following theorem then, shows how we can construct $d$ MUUBs and therefore provides the maximal number possible.

\begin{theorem}
The maximal number of MUUBs for the $d$ dimensional subspace $\mathcal{K}_{\mathcal{S}}$ is $d$.
\end{theorem}
\begin{proof}

\noindent
The isomorphism between the vector space $\mathcal{H}_d$ and $\mathcal{K}_{\mathcal{S}}$ maps between the bases of the vector spaces; The first part of the proof however establishes the fact at least one of the $d+1$ MUBs for $\mathcal{H}_d$ would be mapped to a basis of $\mathcal{K}_{\mathcal{S}}$ which would contain a non-unitary operator.

Beginning with the the computational basis, $\{|0\rangle, ..., |d-1\rangle\}$ for $\mathcal{H}_d$, the $t$-th state from any of the remaining $d$ MUBs can thus be written as 
\begin{eqnarray}\label{ket}
\vert  \gamma_{t}^{(k)} \rangle=\frac{1}{\sqrt{d}} \sum_{h=0}^{d-1} ( \gamma^t)^{d-h} ( \gamma^{-k})^{\alpha_h} \vert h \rangle
\end{eqnarray}
where $\gamma$ is $d$th root of unity, $\gamma=\exp ({2 \pi i /d)}$, $t=0,...,d-1$, $k=0,...,d-1$ (the index $k$ and $t$ indicating the $k$-th basis and element of basis respectively) and $\alpha_h =h+..+(d-1)$ as in Ref. \cite{Durt}. It is instructive to note that $\gamma^i \gamma^j=\gamma^{i \oplus j}$ with $\oplus$ as addition modulo $d$.

Then, consider its conjugate transpose according to definition \ref{defC}, $\vert { \gamma}_{t}^{(k)} \rangle^\dagger$, as 
\begin{eqnarray}
\vert { \gamma}_{t}^{(k)} \rangle^\dagger=\frac{1}{\sqrt{d}} \sum_{h=0}^{d-1} [({ \gamma}^t)^{d-h} ({ \gamma}^{-k})^{\alpha_{(h)}}]^\ast \vert d-h \rangle.
\end{eqnarray}
Thus to determine if $\vert { \gamma}_{t}^{(k)}\rangle$ would be mapped to a unitary operator, we consider the following,
\begin{align}
\vert { \gamma}_{t}^{(k)} \rangle \bullet \vert { \gamma}_{t}^{(k)}\rangle^\dagger
=& \frac{1}{d} \sum_{q=0}^{d-1}\left [ \sum_{h=0}^{d-1} ({ \gamma}^t)^{d-h} (({ \gamma}^t)^{b_h^{(q)}})^\ast \cdot ({ \gamma}^{-k})^{\alpha_{(h)}} (({ \gamma}^{-k})^{\alpha_{(d-b_h^{(q)})}})^\ast\right] \vert q \rangle\nonumber\\
=&\frac{1}{d} \sum_{q=0}^{d-1} \left[ \sum_{h=0}^{d-1} ({ \gamma}^{t(d-h-b_h^{(q)})}) \cdot({\gamma}^{\frac{1}{2}k (d-b_h^{(q)}-(d-b_h)^2+(h-1)h)}) \right]  \vert q \rangle
\end{align}
with the integer $b_h^{(q)}\in[0,d-1]$ such that $h\oplus b_h^{(q)}=q$. We made use of the fact, $\alpha_{(d-b_h^{(q)})}-\alpha_{(h)}= [(d- b_h^{(q)})-(d- b_h^{(q)})^2+(h-1)h]/2$. { As $q=h\oplus b_h^{(q)}$, thus $b_h^{(q)}=q+kd-h$ for some integer $k\ge 0$. As $h$ and $b_h^{(q)}$ are both lesser than $d-1$, thus $k$ cannot be greater than 1.}
With $\gamma^d=\gamma^{b_h^{(q)}d}=1$ and $\gamma^{d/2}=\gamma^{d^2/2}=-1$,  we can further simplify the above into
{\begin{eqnarray}\label{ww}
\vert { \gamma}_{t}^{(k)} \rangle \bullet \vert {\gamma}_{t}^{(k)}\rangle^\dagger
= \frac{1}{d} \sum_{q=0}^{d-1}\left[ \sum_{h=0}^{d-1}{ \gamma}^{hkq}\cdot { \gamma}^{-\frac{k}{2}(q^2+q)-qt } \right] \vert q \rangle.
\end{eqnarray}
Consider the coefficient for the ket $|q\rangle$ is given by 
\begin{eqnarray}\label{cw}
\left[\dfrac{ \gamma^{-\frac{k}{2}(q^2+q)-qt }}{d}\right ]\sum_{h=0}^{d-1}{\gamma}^{hkq}~,
\end{eqnarray}
in} the case of $k=t=0$, we have, the coefficients for every ket (irrespective of $q$) being 1. This implies 
\begin{eqnarray}
\vert { \gamma}_{0}^{(0)} \rangle \bullet \vert { \gamma}_{0}^{(0)} \rangle^\dagger=\sum_{q=0}^{d-1}\vert q \rangle\neq |0\rangle~.
\end{eqnarray}
Thus, the element $\vert { \gamma}_{0}^{(0)} \rangle$ cannot be mapped to a unitary transformation under $G$, hence at least 1 MUB of $\mathcal{H}_d$ cannot be mapped to a unitary basis for $\mathcal{K}_{\mathcal{S}}$.\\
\newline
\noindent The second part of the proof addresses the case for the MUBs with $k \neq 0$; where it is shown that only in the case for $q=0$, the coefficient becomes $1$ (irrespective of $k$) and zero otherwise. The former is straightforward and can be seen by setting $q=0$.
One can rewrite the index of $\gamma$ in the summation of Eq. (\ref{cw}), $hkq$ simply as $hp$ with $p=kq$. It is obvious that $p$ is an integer. Writing $l=hp\mod{d}$, for every integer value of $h\in[0,d-1]$, $l$ will take on a unique integer value in $[0,d-1]$. Thus, for $q \neq 0$,
\begin{equation}
\sum_{h=0}^{d-1} {\gamma}^{(hp)}=\sum_{l=0}^{d-1} { \gamma}^l =0~.
\end{equation}
Thus, $\vert { \gamma}_{t}^{(k)} \rangle \bullet \vert { \gamma}_{t}^{(k)} \rangle^\dagger
=|0\rangle$ for $r\neq 0$. Given the maximal number of MUUBs to be lesser than $d+$, one thus concludes that only $d$ number of MUBs of $\mathcal{H}_d $ are mapped to MUUBs of $\mathcal{K}_{\mathcal{S}}$. 
\end{proof}

\subsection{Recipe for constructing MUUBs for $\mathcal{K}_{\mathcal{S}}$}
\noindent The proof in the previous section is a constructive one, i.e. it provides a way to construct the maximal number of MUUBs for the subspace, $\mathcal{K}_{\mathcal{S}}$. Writing a basis for the subspace as $Z^{(0)}=\{\mathbb{I}_d,Z,Z^2,...,Z^{d-1}\}$, the $t$-th operator of the $k$-th basis, $Z_t^{(k)}$, which is mutually unbiased to $Z^{(0)}$ is given by
\begin{eqnarray}
Z_t^{(k)}=\dfrac{1}{\sqrt{d}} \sum_{i=0}^{d-1} ({ \gamma}^t)^{d-i} ({ \gamma}^{-k})^{\alpha_{(h)}} Z^i 
\end{eqnarray}
with $k=1,...d-1, t=0,...,d-1,$ and $\alpha_{(h)}=h+...+d-1$. { The bases constructed thus are also pairwise mutually unbiased.}
{ We provide an explicit example for the 3-dimensional subspace of $M(3,\mathbb{C})$. Beginning with say a basis, $\mathcal{Z}^{(0)}=\{ \mathbb{I}_3, \mathbb{Z}_3, \mathbb{Z}_3^2\}$ where $\mathbb{Z}_3$ is the generalised Pauli matrix for 3 dimensions, the other two sets of MUUBs, $\{ \mathcal{Z}_0^{(1)}, \mathcal{Z}_1^{(1)}, \mathcal{Z}_2^{(1)} \}$ and $\{ \mathcal{Z}_0^{(2)}, \mathcal{Z}_1^{(2)}, \mathcal{Z}_2^{(2)} \}$ are described by the operators,}
{ \begin{align}
\mathcal{Z}_m^{(1)} &=\frac{1}{\sqrt{3}} [ \mathbb{I}_3 + \gamma^{2m} \mathbb{Z}_3 +\gamma^{m+1} \mathbb{Z}_3^2 ]
\nonumber\\
\mathcal{Z}_n^{(2)} &=\frac{1}{\sqrt{3}} [ \mathbb{I}_3 +\gamma^{2n} \mathbb{Z} _3+\gamma^{n+2} \mathbb{Z}_3^2 ]
\end{align} \newline
\noindent The above gives a total of 3 MUUBs; we note that this corresponds to the result of the numerical search in Ref.\cite{Jesni}.
}

\section{Conclusion}
\noindent
The study of MUUB for $M(d,\mathbb{C})$ and MUBs for $\mathcal{H}_d\otimes\mathcal{H}_d$ consisting of only MES are equivalent; conclusions drawn from one can be used for the other, at least in the context of its construction. Phenomenologically however, they are very different  as one addresses unitary operators acting on quantum states while the other are entangled states. 

As the minimal number of MUUBs that can be constructed is $d(d-1)$ and the case for $d=2$ is known from Ref. \cite{Jesni} to be 3, therefore we can safely conclude that one can construct at least 3 MUUBs for the $d^2$-dimensional space of $M(d,\mathbb{C})$ for any prime number $d$. To date however, no known recipe exists for constructing the maximal number $d^2-1$. Nevertheless, referring to the equivalent problem of MUB for MES, any construction for the $d^2+1$ maximal number of MUBs for a $\mathcal{H}_d\otimes\mathcal{H}_d$, one can consider bases  exclusively consisting either of product states or MES. The minimal number one can achieve in this way for MUBs consisting of MES (equivalently MUUBs) would approach the maximal number, for very large values of $d$.  

In terms of the $d$-dimensional subspace, $\mathcal{K}_{\mathcal{S}}$, i.e. spanned by a basis as that of Eq. (\ref{L33}), the maximal number of MUUB is $d$. A recipe for the construction of such bases is based on an isomorphism between the monoids defined for $\mathcal{H}_d$ and that defined for $\mathcal{K}_{\mathcal{S}}$.

MUUBs suggests an analogous structure for unitary operators as MUBs are for quantum states. With much understanding of it is wanting, immediate directions of study should address the issues of constructing the maximal number of MUUBs possible for $M(d,\mathbb{C})$, equivalences of possible families of MUUBs or even MUUBs for non-prime $d$. In the context of application, MUUBs have played an important role in quantum process tomography and quantum cryptography. In the latter for example, the choice of encoding from differing MUUBs was shown\cite{js} to naturally suppress an eavesdropper information gain. It has also recently been shown to be related to maximal entropic bounds for pairs of setups in distinguishing between unitary processes not unlike the case of MUBs in entropic bounds for observables\cite{ent}. A more thorough study should provide not only further development of the field in its fundamental context, but also its practical implications.

\section{Acknowledgement}

One of the authors, J. S. S. would like to acknowledge financial support under the project FRGS19-141-0750 from the Ministry of Higher Education's Fundamental Research Grant Scheme and the University's Research Management Centre (RMC) for their support and facilities provided.

\end{document}